\documentclass[english,3p,twocolumn]{revtex4-1}
\usepackage[latin9]{inputenc}
\usepackage{amsmath}
\usepackage{amssymb, enumerate}
\usepackage{color}
\usepackage{babel}
\usepackage{graphicx}
\usepackage{epstopdf}
\usepackage{float}
\usepackage{thmtools}
\usepackage{thm-restate}
\usepackage{hyperref}
\usepackage{cleveref}
\usepackage{enumerate}
\usepackage{lipsum}
\usepackage{fixmath}
\makeatletter

\@ifundefined{textcolor}{}
{%
 \definecolor{BLACK}{gray}{0}
 \definecolor{WHITE}{gray}{1}
 \definecolor{RED}{rgb}{1,0,0}
 \definecolor{GREEN}{rgb}{0,1,0}
 \definecolor{BLUE}{rgb}{0,0,1}
 \definecolor{CYAN}{cmyk}{1,0,0,0}
 \definecolor{MAGENTA}{cmyk}{0,1,0,0}
 \definecolor{YELLOW}{cmyk}{0,0,1,0}
}

\newtheorem{theorem}{Theorem}

\newtheorem{proposition}[theorem]{Proposition}

\newtheorem{lem}[theorem]{Lemma}

\newenvironment{proof}[1][Proof]{\noindent\textbf{#1.} }{\ \rule{0.5em}{0.5em}}

\newcommand{\id}{\mathrm{id}}
\newcommand{\bi}{\mathbf{i}}
\newcommand{\bj}{\mathbf{j}}

\DeclareMathOperator{\Tr}{Tr}

\makeatother

\begin{document}

\title{Necessary and sufficient condition of separability for D-symmetric diagonal states}

\author{A. Rutkowski}

\affiliation{
 Institute of Theoretical Physics and Astrophysics,  National Quantum Information Centre\\  Faculty of Mathematics, Physics and Informatics, University of Gda\'{n}sk, 80-952 Gda\'{n}sk, Poland}

\author{M. Banacki and M. Marciniak}

\affiliation{
 Institute of Theoretical Physics and Astrophysics \\  Faculty of Mathematics, Physics and Informatics, University of Gda\'{n}sk, 80-952 Gda\'{n}sk, Poland}

\begin{abstract}
For multipartite states we consider a notion of D-symmetry. For a system of $N$ qubits it concides with usual permutational symmetry. In case of $N$ qudits ($d\geq 3$) the D-symmetry is stronger than the permutational one.
For the space of all D-symmetric vectors in $(\mathbb{C}^d)^{\otimes N}$ we define a basis composed of vectors $\{|R_{N,d;k}\rangle: \,0\leq k\leq N(d-1)\}$ which are analog of Dicke states. 
The aim of this paper is to discuss  the problem of separability of D-symmetric states which are diagonal in the basis $\{|R_{N,d;k}\rangle\}$. We show that if $N$ is even and $d\geq 2$ is arbitrary then a PPT property is necessary and sufficient condition of separability for D-invariant diagonal states. In this way we generalize results obtained in \cite{Yu2016,WoYe2014}. 
Our strategy is to use some classical mathematical results on a moment problem \cite{KrNu1977}.

\end{abstract}



\maketitle


\section{Introduction}
Quantum theory is the primary mainstay of our understanding and formal description of Nature.
The phenomenon of quantum correlations, especially entanglement, is believed to be most amazing and eluding the schemes of classical thinking.  Multiannual conceptual efforts to grapple the "spooky actions for separated distance systems" began with the fundamental work of Einsten, Podolski and Rosen \cite{EPR} and continue until this day. Despite of the knowledge we possess today, quantum correlations still remain a great mystery.  In particular, it is not easy to recognize the type of correlation we work with, even in a bipartite system.
There exist just a few criteria to detect entanglement. The most famous is Peres-Horodecki criterion \cite{HHH1996, Peres1996} which states that if state is non PPT then it is entanglement (PPT-positive partial transposition criteria). It turns out that this is necessary and sufficient condition only for the low dimensional systems, i.e. $\mathbb{C}^2\otimes\mathbb{C}^2$ and $\mathbb{C}^2\otimes\mathbb{C}^3$.

In last decade the problem of separability of permutationally symmetric states was intensively analyzed \cite{ToGu2009,Toth2010a,Toth2010b,WoYe2014,Yu2016}. The question how to define generalized Dicke states for qudits arises. Some authors \cite{Tura2018} consider a naturally choosen basis for the bosonic subspace of $(\mathbb{C}^d)^{\otimes N}$ in this context. We propose another approach. Instead of the full bosonic subspace of permutationally symetric vectors in the tensor product we define a smaller subspace of D-symmetric vectors. It is defined as the image of a D-symmetrizator $P_\mathrm{D}$ which has the property $P_\mathrm{D}P_\mathrm{S}=P_\mathrm{S}P_\mathrm{D}=P_\mathrm{D}$, where $P_\mathrm{S}$ stands for the usuall symmetrizator. It follows that for $d=2$ the space of D-symmetric vectors is nothing but the bosonic subspace of $(\mathbb{C}^d)^{\otimes N}$. Moreover, it is possible to define a basis of the D-symmetric subspace which for $d=2$ coincides with the basis of Dicke states.  
In next section the details of our construction can be found as well as some possible physical motivation for this. 

It was observed by several authors \cite{Toth2010b,WoYe2014} that there is a strong connection between separability and PPT property for mixtures of Dicke states. It was even conjectured that PPT property is sufficient for separability of such states. It turns out to be true for qubits \cite{Yu2016}. For general qudits, when one considers the basis of bosonic subspace, it is no longer true \cite{Tura2018}. The aim of our paper is to show that the conjecture is still true when one considers the basis of D-symmetric states as a generalization of Dicke states.

The paper is organized as follow. Firstly, the notion of D-symmetry is discussed (Section \ref{ss:D-inv}). We also define a D-symmetric analog of Dicke states. In Section \ref{s:separability} we recall the seprability problem and formulate the appropriate notion of PPT property in the context of multipartite systems. Section \ref{s:D-sep} is devoted to characterization of D-symmetric separable states, while in Section \ref{s:D-wit} we provide a description of entanglement witnesses for D-symmetric systems. Next, we formulate the generalized moment problem and recall a description of the complete solution of it (Section \ref{s:moment}). In Section \ref{s:PPT} we provide conditions characterizing diagonal restricted Dicke states satisfying PPT property, while in Section \ref{s:sepDicke} we characterize separable diagonal Dicke states. Finally, in Section \ref{s:sepPPT} we formulate main theorem (Theorem 9) which states for an even number of qudits sperability of diagonal restricted Dicke states is equivalent to PPT property with respect to the half of subsystems.


\section{D-symmetry of multipartite states}
\label{ss:D-inv}
Let $d\geq 2$ and $N\geq 2$ be fixed numbers.
For $0\leq i_1,\ldots,i_N\leq d-1$ the $N$-tuple $(i_1,\ldots,i_N)$ will be denoted shortly by $\mathbold{i}$.  
Let $|\mathbold{i}|=i_1+\ldots+i_N$. For an $N$ tuple $\mathbold{i}$ and a number $k\in\{0,1,\ldots, N(d-1)\}$ we will write $\mathbold{i}\vdash k$ when $|\mathbold{i}|=k$.
We define
$
{N\choose k}_d=\#\{\mathbold{i}:\,\mathbold{i}\vdash k\}
$ for $k=0,1,\ldots N(d-1)$. 
Observe that ${N\choose k}_2={N\choose k}$ where ${N\choose k}$ stands for the usual Newton binomial coefficient. There is no simple formula for ${N\choose k}_d$ for $d>2$. However, the following recurrence formula is valid
\begin{equation}
{N \choose k}_d=\sum_{j=0}^{\min\{ k,d-1 \}  } {N-1 \choose k-j}_d
\end{equation}
This is a generalization of the basic formula determining the Pascal triangle for $d=2$.

For an $N$-tuple $\mathbold{i}=(i_1,\ldots,i_N)$ let $|\mathbold{i}\rangle=|i_1,\ldots,i_N\rangle$ be an element of the standard orthonormal basis in $(\mathbb{C}^d)^{\otimes N}$. By $P_\mathrm{S}$ we denote standard symmetrizator acting on $(\mathbb{C}^d)^{\otimes N}$, i.e. a projection defined by 
$P_S|\mathbold{i}\rangle = (N!)^{-1}\sum_{\sigma\in S_N}|\sigma(\mathbold{i})\rangle$ where $S_N$ denotes the group of permutations of the set $\{1,2,\ldots,N\}$ and $\sigma(\mathbold{i})=(i_{\sigma(1)},i_{\sigma(2)},\ldots,i_{\sigma(N)})$ for $\sigma\in S_N$. The image of $P_\mathrm{S}$ is the bosonic subspace of $(\mathbb{C}^d)^{\otimes N}$ and it is denoted by $\left((\mathbb{C}^d)^{\otimes N}\right)_{\mathrm{sym}}$. For $d=2$ one considers the basis $(|D_{N,k}\rangle)_{k=0,1,\ldots,N}$ of $\left((\mathbb{C}^2)^{\otimes N}\right)_{\mathrm{sym}}$, where
\begin{equation}
\label{Dicke}
|D_{N;k}\rangle={N\choose k}P_\mathrm{S}\left(|0\rangle^{\otimes (N-k)}\otimes |1\rangle^{\otimes k}\right) .
\end{equation}
Following \cite{WoYe2014,Yu2016} the elements of the basis are called (unnormalized) Dicke states.
This notion was generalized in \cite{Tura2018} for an arbitrary dimension $d\geq 2$. In this case Dicke states are parametrized by systems $(k_0,\ldots,k_{d-1})$ of nonnegative integers such that $k_0+k_1+\ldots +k_{d-1}=N$, and
\begin{eqnarray}
\lefteqn{|D_{N,d;k_0,\ldots,k_{d-1}}\rangle=}\\
&=&{N\choose k_0,\ldots,k_{d-1}}P_\mathrm{S}\left(|0\rangle^{\otimes k_0}\otimes\ldots\otimes|d-1\rangle^{\otimes k_{d-1}}\right) \nonumber
\end{eqnarray}


We will consider a subspace $\left((\mathbb{C}^d)^{\otimes N}\right)_{\mathrm{D}}\subset\left((\mathbb{C}^d)^{\otimes N}\right)_{\mathrm{sym}}$ which is defined as an image of a projection $P_\mathrm{D}$ defined by
\begin{equation}
\label{PD}
P_\mathrm{D}|\mathbold{i}\rangle = {N\choose |\bi|}_d^{-1}\sum_{{\mathbold{j} \vdash|\mathbold{i}|}}|\mathbold{j}\rangle ,
\end{equation}
where $|\mathbold{i}\rangle$ is any element of the standard basis of $(\mathbb{C}^d)^{\otimes N}$. 
Obviously 
\begin{equation}
\label{commut}
P_\mathrm{D}P_\mathrm{S}=P_\mathrm{S}P_\mathrm{D}=P_\mathrm{D}.
\end{equation}
For fixed $N$ and $d$, we consider a basis of $\left((\mathbb{C}^d)^{\otimes N}\right)_{\mathrm{D}}$ composed of vectors $|R_{N,d;k}\rangle$, $k=0,1,2,\ldots, N(d-1)$ defined by 
\begin{equation}
|R_{N,d;k}\rangle = 
\sum_{\bi\vdash k} |\bi \rangle
\end{equation}
Elements $|R_{N,d;k}\rangle$ will be called restricted (unnormalized) Dicke states. One can easily observe that $|R_{N,2;k}\rangle=|D_{N,k}\rangle$ where $|D_{N,k}\rangle$ are defined as in \eqref{Dicke}.  
Therefore, restricted Dicke states can be regarded as a generalization of $N$-qubit Dicke states, which is  diffrent from \cite{Tura2018}.

By $|\widetilde{R_{N,d;k}}\rangle$ we will denote elements of the dual basis to the basis of restricted Dicke states. One easily checks that
\begin{equation}
\label{dual}
|\widetilde{R_{N,d;k}}\rangle={N \choose k}_d^{-1}\sum_{\bi\vdash k} |\bi \rangle={N\choose k}_d^{-1}|R_{N,d:k}\rangle.
\end{equation}

Assume that a system is composed of $N$ bosons with $d$ levels of excitation each. We make an assumption that subsequent levels differ by a fixed value. Then $|R_{N,d;k}\rangle$ can be interpreted as such a state of the system that the total number of excitations in all bosons is equal to $k$. It can be used to model systems of bosons concentarted in a small area which behave as single particle and only total energy can be recognized. Such models were used to explain the notion of superradiance in quantum optics \cite{Dicke,GH}

By $\mathfrak{D}$ we denote the class of states $\rho$ which satisfy the condition $\rho= P_\mathrm{D}\rho P_\mathrm{D}$. We will address this property as D-symmetry of the state $\rho$. 
It is stronger than the permutational symmetry in the sense that each D-symmetric state is automatically a permutationally symmetric one.

Special attention is put on a class of diagonal Dicke (unnormalized) states \cite{WoYe2014,Yu2016,Tura2018,QRS}, i.e. states of the form
$ 
\rho
=\sum_{\{k_j\}
}p_{k_0,\ldots,k_{d-1}}|D_{N,d;k_0,\ldots,k_{d-1}}\rangle\langle D_{N,d;k_0,\ldots,k_{d-1}}| ,
$ 
where the summation is over all systems $\{k_j\}$ of integers such that $j=0,1,\ldots d-1$, $k_j\geq 0$,   $\sum_jk_j=N$, and numbers $p_{k_0,\ldots,k_{d-1}}$ are nonnegative.

The natural analogs among D-symmetric states are diagonal restricted Dicke states which are of the form
\begin{equation}
\label{rDd}
\rho=\sum_{k=0}^{N(d-1)}p_k|R_{N,d;k}\rangle\langle R_{N,d;k}|
\end{equation}
for $p_k\geq 0$.

\section{Separability problem}
\label{s:separability}

%
%
We briefly recall main notions and concepts concerning the separability problem of states. Let $\mathcal{H}=\mathcal{H}_1\otimes \ldots \otimes \mathcal{H}_N$ be a composite system where each subsystem is represented by a Hilbert space $\mathcal{H}_j$, $j=1,2,\ldots,N$. Let  $\rho$ be an (unnormalized) state i.e. $\rho\in\mathcal{B}(\mathcal{H})$ is a positive semidefite operator. We say that $\rho$ is fully separable (or shortly separable)  if $\rho=\sum_{i}p_{i}\rho_{i}^1\otimes \ldots \otimes\rho_{i}^N$,
for some states $\rho_{i}^j$  on $\mathcal{H}_j$, $j=1,\ldots,N$, and some positive numbers $p_{i}$. 
Whenever $\rho$ does not satisfy this condition it said to be an entangled state.
In general, it is very difficult to check whether a state is separable.

We say that a Hermitian operator $W$ on $\mathcal{H}_1\otimes\ldots\mathcal{H}_N$ is an entanglement witness if $$\langle \psi_1\otimes\ldots\otimes \psi_N|W|\psi_1\otimes\ldots\otimes \psi_N\rangle\geq 0$$ for every $\psi_1,\ldots,\psi_N$ such that $\psi_j\in\mathcal{H}_j$, but $W$ is not a positive operator. A state $\rho$ is entangled if and only if $\Tr(\rho W)<0$ for some entanglement witness $W$. In this case we say that $W$ detects the entanglement of $\rho$.

Let $T_j$ denote that transposition on the algebra $\mathcal{B}(\mathcal{H}_j)$. A bipartite state $\rho$ is said to posses a PPT property if $(T_1\otimes \mathrm{id}_{\mathcal{H}_2})\rho$ is a also a state. We generalize this property to the multipartite case. Let $(m_1,\ldots,m_N)\in\{0,1\}^N$ be a binary system of the lenght $N$. We say that a state $\rho$ on $\mathcal{H}_1\otimes\ldots\otimes\mathcal{H}_N$ has a $(m_1,\ldots,m_N)$-PPT property if $(T_1^{m_1}\otimes\ldots\otimes T_N^{m_N})\rho$ is also a state, where $T_j^0=\mathrm{id}_j$ and $T_j^1=T_j$, i.e. all $1$'s in the system $(m_1,\ldots,m_N)$ mark subsytems which are transposed. Clearly, if a state $\rho$ is separable then it has a $(m_1,\ldots,m_n)$-PPT property for every binary system $(m_1,\ldots,m_n)$. In general, the converse implication is not true unless $N=2$ and the pair $(\mathcal{H}_1,\mathcal{H}_2)$ is one of the following: $(\mathbb{C}^2,\mathbb{C}^2),  (\mathbb{C}^2,\mathbb{C}^3), (\mathbb{C}^3,\mathbb{C}^2)$.

In spite of this general statement there are classes of states such that PPT property implies separability within them. For example, as was shown in \cite{Yu2016} the class of diagonal Dicke states for $d=2$ has this property. It was conjectured in \cite{WoYe2014}. A natural question arises whether the same is for the set of Dicke states for any dimension $d$ of the underlying Hilbert space. It turns out that for $d>2$ it is not the case \cite{Tura2018}. 
The aim of our paper is to show that for the class of diagonal restricted Dicke states PPT property implies separability.

\section{Separable D-symmetric states}
\label{s:D-sep}
The aim of this section is to characterize separable D-symmetric states.

Let $\rho$ be a separable state of the form 
\begin{equation}
\label{ropure}
\rho=\sum_{\alpha=1}^n\lambda_\alpha p_\alpha^1\otimes\ldots\otimes p_\alpha^N,
\end{equation}
where $p_\alpha^i=|\xi_\alpha^i\rangle\langle\xi_\alpha^i|$, $i=1,\ldots,N$, $\alpha=1,\ldots,n$, for some vectors $\xi_\alpha^i\in H$ such that $\Vert\xi_\alpha^i\Vert=1$. 
The result established in the next proposition is probably well known. However, we didn't find any reference with a complete proof. Hence, we provide the proof for the readers convenience.
\begin{proposition}
\label{sepper}
Assume that $\rho$ given by \eqref{ropure} is permutationally symmetric, i.e.
$ 
\rho = P_\mathrm{S}\rho P_\mathrm{S}.
$ 
If all coefficients $\lambda_\alpha$ in \eqref{ropure} are strictly positive then $p_\alpha^i=p_\alpha^j$ for every $\alpha=1,\ldots,n$ and $i,j=1,\ldots,N$.
\end{proposition}
\begin{proof}
Let $\rho_\alpha=p_\alpha^1\otimes\ldots\otimes p_\alpha^N$. 
Observe that for a projection $P$, a selfadjoint operator satisfies $A=PAP$ is and only if $\langle\eta,A\eta\rangle=0$ for every $\eta\in (PH)^\perp$. 
So, if $\eta\in\left((\mathbb{C}^d)^{\otimes N}\right)_{\mathrm{sym}}^\perp$ then $\langle \eta,\rho\eta\rangle=0$, and consequently $\sum_\alpha\lambda_\alpha\langle\eta,\rho_\alpha\eta\rangle=0$. Each term in the sum is nonnegative, hence $\lambda_\alpha\langle\eta, \rho_\alpha\eta\rangle=0$ for every $\alpha$. Since $\lambda_\alpha$ are nonzero, $\langle\eta, \rho_\alpha\eta\rangle=0$. 

Now, fix $\alpha$ and a pair $i,j$ where $1\leq i<j\leq N$, and let
\begin{eqnarray*}
\eta_\alpha^{i,j}&=&
|\xi_\alpha^1,\ldots,\xi_\alpha^i,\ldots,\xi_\alpha^j,\ldots,\xi_\alpha^N \rangle\\
&& {}
- |\xi_\alpha^1,\ldots,\xi_\alpha^j, \ldots, \xi_\alpha^i, \ldots , \xi_\alpha^N\rangle.
\end{eqnarray*}
Then $\eta_\alpha^{i,j}\in\left((\mathbb{C}^d)^{\otimes N}\right)_{\mathrm{sym}}^\perp$, and
$ 
0=\langle \eta_\alpha^{i,j},\rho_\alpha \eta_{\alpha}^{i,j} \rangle= 1-|\langle\xi_\alpha^i,\xi_\alpha^j\rangle|^2
$. 
The equality $|\langle\xi_\alpha^i,\xi_\alpha^j\rangle|=1$ implies that $\xi_\alpha^i$ and $\xi_\alpha^j$ are linearly dependent, so $|\xi_\alpha^i\rangle\langle\xi_\alpha^i|=|\xi_\alpha^j\rangle\langle\xi_\alpha^j|$.
\end{proof}

Now, assume that the state $\rho$ given by \eqref{ropure} is D-symmetric. Due to \eqref{commut} the state $\rho$ is permutationally symmetric. Hence, by Proposition \ref{sepper}, it is of the form
\begin{equation}
\label{rosymm}
\rho=\sum_\alpha \lambda_\alpha |\xi_\alpha\rangle\langle\xi_\alpha|^{\otimes N}
\end{equation}
for some vectors $\xi_\alpha\in H$, $\alpha=1,\ldots,n$. We show that D-symmetry implies a very special form of vectors $\xi_\alpha$.

\begin{proposition}
\label{p:Dpure}
Assume that $\rho$ given by \eqref{rosymm} is D-symmetric, i.e.
$
\rho = P_\mathrm{D}\rho P_\mathrm{D}.
$ 
Then for each $\alpha=1,\ldots,n$, either
\begin{equation}
\label{Dsep1}
|\xi_\alpha\rangle=|d-1\rangle
\end{equation}
 or there is a number $z
\in\mathbb{C}$ such that 
\begin{equation}
\label{Dsep2}
|\xi_\alpha\rangle= C_{z
}\sum_{i=0}^{d-1}z
^i|i\rangle,
\end{equation}
where $C_z=\left(\sum_{i=0}^{d-1}|z|^2\right)^{-1/2}$ for $z\in\mathbb{C}$.
\end{proposition}
\begin{proof}
Arguing as in the proof of Proposition \ref{sepper} one shows $|\xi_\alpha\rangle\langle\xi_\alpha|^{\otimes N}=P_\mathrm{D} |\xi_\alpha\rangle\langle\xi_\alpha|^{\otimes N} P_\mathrm{D}$ for every $\alpha$.
Fix $\alpha$ and let $|\xi_\alpha\rangle=\sum_{i=0}^{d-1} w_i|i\rangle$. Then
\begin{equation}
\rho=\sum_{\bi,\bj}w_\bi\overline{w_\bj}|\bi\rangle\langle\bj|
\end{equation}
where $w_\bi=w_{i_1}w_{i_2}\ldots w_{i_N}$ for an $N$-tuple $\bi=(i_1,\ldots,i_N)$, and the summation is over all pairs of $N$-tuples $\bi,\bj$. 

We show that D-symmetry implies  
the following condition:
\begin{equation}
\label{for}
\forall \bi,\bi':\;|\bi|=|\bi'|\quad\Rightarrow\quad \xi_\bi=\xi_{\bi'} .
\end{equation}
It follows from \eqref{PD} that
$$P_\mathrm{D}|\xi_\alpha\rangle\langle\xi_\alpha|^{\otimes N} P_\mathrm{D}=
\sum_{\bi,\,\bj} y_{\bi,\bj} |\bi\rangle\langle\bj|
$$
where coefficients
$$y_{\bi,\bj}={N\choose |\bi|}_d^{-1} {N\choose |\bj|}_d^{-1}   \sum_{\substack{\bi', \bj' \\ |\bi'|=|\bi|,\, |\bj'|=|\bj|}} w_{\bi'}\overline{w_{\bj'}}$$
depend on numbers $|\bi|$ and $|\bj|$ only. 
Since $|\xi_\alpha\rangle\langle\xi_\alpha|^{\otimes N}=P_\mathrm{D}|\xi_\alpha\rangle\langle\xi_\alpha|^{\otimes N} P_\mathrm{D}$, for any $N$-tuples  $\bi$, $\bi'$, $\bj$, $\bj'$, the equalities $|\bi|=|\bi'|$ and $|\bj|=|\bj'|$ imply $w_\bi\overline{w_\bj}=w_{\bi'}\overline{w_{\bj'}}$. 
If $w_\bi=0$ for every $\bi$ then the condition \eqref{for} is satisfied. Assume that $w_{\bi_0}\neq 0$ for some $N$-tuple $\bi_0$. If $|\bi|=|\bi'|$, then $w_{\bi}\overline{w_{\bi_0}}=w_{\bi'}\overline{w_{\bi_0}}$, hence $w_{\bi}=w_{\bi'}$ and \eqref{for} is proved.

Now, consider two cases. Firstly, assume $w_0=0$. Condition \eqref{for} implies equality $w_i^N=w_i^{N-2}w_{i+1}w_{i-1}$ for every $i=1,\ldots, d-2$. Hence, by induction one shows that $w_i=0$ for every $i=0,1,\ldots, d-2$, and \eqref{Dsep1} follows. Secondly, consider the case $w_0\neq 0$. Again from \eqref{for} we get $w_iw_0^{N-1}=w_{i-1}w_1w_0^{N-2}$ for every $i=1,2,\ldots, d-1$. Hence $w_i=zw_{i-1}$ for $i=1,2,\ldots,d-1$, where $z=w_1/w_0$, and \eqref{Dsep2} follows.
\end{proof}
%

\section{Entanglement witnesses for D-symmetric systems}
\label{s:D-wit}
In this subsection we deal with the approach to entanglement witnesses for a bosonic systems presented in \cite{Yu2016}.  We say that a Hermitian operator $W\in B((\mathbb{C}^d)^{\otimes N})$ is an entaglement witness for the D-symmetric system if $W=P_\mathrm{D}WP_\mathrm{D}$ and $\Tr(W\sigma)\geq 0$ for all pure separable D-symmetric states. Observe that due to results from previous section every pure separable D-symmetric state is of the form
\begin{equation}
\sigma_z=C_z^N \sum_{\bi,\bj}z^{|\bi|}\overline{z}^{|\bj|}|\bi\rangle\langle\bj|
\end{equation}
for some $z\in\mathbb{C}$.

The following proposition is a simple consequence of the hyperplane separation theorem (see  \cite[Proposition 1]{Yu2016}). 
\begin{proposition}
\label{p:ew}
A D-symmetric state $\rho$ is separable if and only if $\Tr(W\rho)\geq 0$ for every entanglement witness $W$ for the D-symmetric system.
\end{proposition}

Let us remind that  $|\widetilde{R_{N,d;k}}\rangle$ denote dual vectors to $|R_{N,d;k}\rangle$, see \eqref{dual}.
\begin{proposition}
\label{DEW}
Let $n_1=\left\lfloor \frac{N(d-1)}{2}\right\rfloor$ and $n_2=\left\lfloor\frac{N(d-1)-1}{2}\right\rfloor$.
Let two systems $(s_k)_{0\leq k\leq n_1}$ and $(t_k)_{0\leq k\leq n_2}$ of complex numbers be given. Define
\begin{equation}
\label{Ws}
V_{(s)}=\sum_{k,l=0}^{n_1}s_k\overline{s_l}|\widetilde{R_{N,d;k+l}}\rangle\langle\widetilde{R_{N,d;k+l}}|
\end{equation}
\begin{equation}
\label{Wt}
U_{(t)}=\sum_{k,l=0}^{n_1}t_k\overline{t_l}|\widetilde{R_{N,d;k+l+1}}\rangle\langle\widetilde{R_{N,d;k+l+1}}|
\end{equation}
Then $V_{(s)}$ and $U_{(t)}$ are entanglement witnesses for D-symmetric systems.
\end{proposition}
We provide the proof of the above Proposition in Appendix A. 

\section{Moment problem}
\label{s:moment}
Herein, we will recall the concept of moment problem. It has been shown that methods of the generalized moment problem are useful in geometry of convex bodies, algebra and function theory (see for example \cite{KrNu1977}). 
Let $(p_k)_{k=0}^n$ be a finite sequence of real numbers. We say that the sequence $(p_k)$ is a solution of the generalized moment problem on the interval $[0,\infty)$ \cite{KrNu1977} if there exists a positive measure $\sigma$ with support contained in $[0,\infty)$ such that
\begin{equation}
\label{e:moment}
p_k=\begin{cases}
{\displaystyle \int_0^\infty t^k d\sigma(t)}, & k=0,1,\ldots,n-1, \\[4mm]
{\displaystyle \int_0^\infty t^n d\sigma(t) + M}, & k=n,
\end{cases}
\end{equation}
where $M\geq 0$.
Alternatively, we say that it is a solution of the strict moment problem on the interval $[0,\infty)$ if it is a solution of the generalized moment problem with $M=0$. The following theorem characterizes solutions of the generalized moment problem completely. It will be our main mathematical tool.  
\begin{theorem}[Chapter V in \cite{KrNu1977}]\label{t:moment}
A sequence $(p_k)_{k=0}^n$ is a solution of the generalized moment problem if and only if the following two Henkel matrices
\begin{equation*}
(p_{k+l})_{k,l=0}^{n_0}=\left(\begin{array}{ccccc}
p_0 & p_1 & p_2 & \cdots & p_{n_0} \\
p_1 & p_2 & p_3 & \cdots & p_{n_0+1} \\
p_2 & p_3 & p_4 & \cdots & p_{n_0+2} \\
\vdots & \vdots & \vdots & & \vdots \\
p_{n_0} & p_{n_0+1} & p_{n_0+2} & \cdots & p_{2n_0}
\end{array}\right),
\end{equation*}	
\begin{equation*}
(p_{k+l+1})_{k,l=0}^{n_1}=\left(\begin{array}{ccccc}
p_1 & p_2 & p_3 & \cdots & p_{n_1+1} \\
p_2 & p_3 & p_4 & \cdots & p_{n_1+2} \\
p_3 & p_4 & p_5 & \cdots & p_{n_1+3} \\
\vdots & \vdots & \vdots & & \vdots \\
p_{n_1+1} & p_{n_1+2} & p_{n_1+3} & \cdots & p_{2n_1+1}
\end{array}\right)
\end{equation*}	
are positive semidefinite, where $n_0=\lfloor \frac{n}{2}\rfloor$ and $n_1=\lfloor \frac{n-1}{2}\rfloor $. If both matrices are strictly positive definite then the sequence is a solution of the strict moment problem.
\end{theorem}

\section{Characterization of PPT diagonal restricted Dicke states} 
\label{s:PPT}
As we observed in section \ref{s:separability}, in the multipartite case there is a family of conditions which are called PPT properties. They are indexed by the set of binary systems $(m_1,\ldots,m_N)$. One can check that in case of permutationally symetric states all $(m_1,\ldots,m_N)$-PPT conditions with fixed $m:=m_1+m_2+\ldots+m_N$, i.e. number of subsystems which are transposed, are equivalent (see Appendix B. 
Thus, it is enough to consider only PPT conditions with first $m$ subsystems transposed, where $m\leq\lfloor\frac{N}{2}\rfloor$.

We start with the following characterization of $m$-PPT property in terms of the associated Henkel matrices.
\begin{theorem}\label{t:PPT}
Let $m\leq \frac{N}{2}$. The state $\rho$ is $m$-PPT if and only if
\begin{enumerate}[(a)]
\item matrices $(p_{i+j})_{i,j=0}^{m(d-1)}$ and $(p_{i+j+1})_{i,j=0}^{m(d-1)-1}$ are positive definite, when $N=2m$,
\item matrices $(p_{i+j+l})_{i,j=0}^{m(d-1)}$, $l=0,\ldots,(N-2m)(d-1)$, are positive definite, when $2m<N$.
\end{enumerate}
\end{theorem}
\begin{proof}
According to \eqref{rDd} one has
\begin{equation}
\label{e:rD}
\rho=\sum_{k=0}^{N(d-1)}p_k\;\sum_{\bi\vdash k}\; \sum_{\bj\vdash k} |\bi\rangle\langle\bj|.
\end{equation}
Let $\Gamma_m=T^{\otimes m}\otimes \id^{\otimes (N-m)} $ be the partial transposition with respect to $m$ first subsystems.
Then one can show (see Appendix C) 
that 
\begin{equation}
\label{e:Al}
\Gamma_m(\rho)=
\sum_{s=-m(d-1)}^{(N-m)(d-1)} A_s
\end{equation}
where 
\begin{widetext}
\begin{equation}
\label{e:As}
A_s= \sum_{k,l=\max\{0,-s\}}^{\min\{m(d-1),(N-m)(d-1)-s\}} p_{k+l+s}|R_{m,d;k},R_{N-m,d;k+s}\rangle \langle R_{m,d;l},R_{N-m,d;l+s}|.
\end{equation}
\end{widetext}
One can observe that $A_s$ are hermitian and $A_sA_{s'}=0$ for $s\neq s'$. Thus, $\Gamma_m(\rho)$ is positive definite if and only if each $A_s$ is positive definite. It is equivalent to positive definiteness of the following Hecke matrices 
\begin{equation}
\nonumber 
P_s=(p_{k+l+s})_{\max\{0,\,-s\}\leq k,\,l\leq \min\{m(d-1),\,(N-m)(d-1)-s\}}
\end{equation}
for all $s\in \{-m(d-1),-m(d-1)+1,\ldots, (N-m)(d-1)\}$.

If $A$ and $B$ are Hermitian matrices, then we will write $A\subset B$ if $\dim A\leq \dim B$, and $A$ is a principal submatrix of $B$. Obviously, if $A\subset B$ and $B$ is positive definite then $A$ is positive definite too.
 
Firstly, let us consider the case $2m=N$. We will show that if the matrices
$P_0=(p_{k+l})_{k,l=0}^{m(d-1)}$ 
and
$P_1=(p_{k+l+1})_{k,l=0}^{m(d-1)-1}$ 
are positive definite then $P_s$ are positive definite for all $s$ such that $-m(d-1)\leq s\leq m(d-1)$.
Assume that $s$ is any even number, i.e. $s=2q$. Then
\begin{equation}
\nonumber
P_{2q}= \begin{cases}
(p_{k+l+2q})_{k,l=0}^{m(d-1)-2q} & \mbox{if $q>0$,} \\[4mm]
(p_{k+l-2|q|})_{k,l=2|q|}^{m(d-1)} & \mbox{if $q<0$.}
\end{cases} 
\end{equation}
One observes that in both cases $P_{2q}=(p_{k+l})_{k,l=|q|}^{m(d-1)-|q|}\subset P_0$, hence $P_{2q}$ is positive definite.
For $s=2q+1$, we have
$$ 
P_{2q+1}=
(p_{k+l+2q+1})_{k,l=0}^{m(d-1)-2q-1} =(p_{k+l+1})_{k,l=q}^{m(d-1)-q-1}
$$ 
if $q\geq 0$, and
$$ 
P_{2q+1}=
(p_{k+l-2|q|+1})_{k,l=2|q|+1}^{m(d-1)} = (p_{k+l+1})_{k,l=|q|+1}^{m(d-1)-|q|}
$$ 
if $q<0$.
In both cases $P_{2q+1}\subset P_1$. 

Now, let us consider the case $2m<N$. Assume that $P_s$ are positive definite for all $s$ such that $0\leq s\leq (N-2m)(d-1)$. For $s$ such that $-m(d-1)\leq s <0$, one can use the same arguments as above 
to show that $P_s\subset P_0$ or $P_s\subset P_1$. If $(N-2m)(d-1)<s\leq (N-m)(d-1)$ then $s=(N-2m)(d-1)+2q$ or $l=(N-2m)(d-1) +2q+1$. One has
\begin{eqnarray*}
\lefteqn{
P_{(N-2m)(d-1)+2q} =}\\ & = &
(p_{k+l+(N-2m)(d-1)+2q})_{k,l=0}^{m(d-1)-2q}\\
&=&(p_{k+l+(N-2m)(d-1)})_{k,l=q}^{m(d-1)-q}
\subset P_{(N-2m)(d-1)}
\end{eqnarray*}
and
\begin{eqnarray*}
\lefteqn{
P_{(N-2m)(d-1)+2q+1}=} \\ & = & 
(p_{k+l+(N-2m)(d-1)+2q+1})_{k,l=0}^{m(d-1)-2q-1}\\ 
&=&(p_{k+l+(N-2m)(d-1)-1})_{k,l=q+1}^{m(d-1)-q}\subset P_{(N-2m)(d-1)-1}
\end{eqnarray*}
Thus the proof is complete.
\end{proof}


\section{Separability of diagonal restricted Dicke states vs. moment problem}
\label{s:sepDicke}
We start with the following observation.
\begin{lem}
\label{p:geom}
If $(p_k)_{k=0,1,\ldots,N(d-1)}$ is a geometric sequence then $\rho$ given by \eqref{rDd} is fully separable.
\end{lem}
\begin{proof}
Let $p_k=t^k$ for $k=0,1,\ldots,N(d-1)$ and for some $t>0$. By $\omega$ we denote the $N(d-1)$-th root of $1$, i.e. $\omega=\exp\dfrac{2\pi i}{N(d-1)+1}$.	
Consider the following "$t$-dual" vectors to the computational basis with respect to the discrete Fourier transform
\begin{equation}
|\hat{\alpha}\rangle = \sum_{i=0}^{d-1}t^{i/2} \omega^{\alpha i}|i\rangle, \qquad \alpha=0,1,\ldots,N(d-1).
\end{equation}
Now, observe that
\begin{eqnarray*}
\lefteqn{\dfrac{1}{N(d-1)+1}\sum_{\alpha=0}^{N(d-1)}|\hat{\alpha}\rangle\langle\hat{\alpha}|^{\otimes N}=} \\
& = & 
\dfrac{1}{N(d-1)+1}\sum_{\alpha=0}^{N(d-1)}\;\sum_{\bi,\bj}t^{\frac{1}{2}(|\bi|+|\bj|)}\,\omega^{\alpha(|\bi|-|\bj|)}|\bi\rangle\langle \bj| \\
&=&
\dfrac{1}{N(d-1)+1}\sum_{\bi,\bj} t^{\frac{1}{2}(|\bi|+|\bj|}
\sum_{\alpha=0}^{N(d-1)}\left(\omega^{|\bi|-|\bj|}\right)^\alpha
|\bi\rangle\langle\bj| \\
&=&
\sum_{\bi,\bj} \delta_{|\bi|,|\bj|}\; t^{\frac{1}{2}(|\bi|+|\bj|} |\bi\rangle\langle\bj| 
= 
\sum_{\substack{\bi,\bj \\ |\bi|=|\bj|}} t^{|\bi|}|\bi\rangle\langle \bj| \\
&=&
\sum_{k=0}^{N(d-1)}\; t^k \;\sum_{\bi\vdash k}\; \sum_{\bj\vdash k} |\bi\rangle\langle \bj|=
\sum_{k=0}^{N(d-1)} t^k \, |R_{N,d;k}\rangle\langle R_{N,d;k}|\\
& = &\rho
\end{eqnarray*}
Thus $\rho$ is separable.
\end{proof}

We are in the position to formulat the following characterization 
\begin{theorem}
\label{t:Narb}
Let $d\geq 2$ and $N$ be arbitrary.
The state $\rho$ given by \eqref{rDd} is fully separable if and only the sequence $(p_k)_{k=0}^{N(d-1)}$ is a solution of the generalized moment problem.
\end{theorem}

\begin{proof}
\textit{Necessity.}
If $\rho$ is separable then $\Tr(\rho W)\geq 0$ for every entanglement witness for D-symmetric systems.
For a system $(s_k)_{0\leq k\leq \lfloor N(d-1)/2\rfloor}$ of complex numbers let us consider the entanglemnet witness $V_{(s)}$ given by \eqref{Ws} in Proposition \ref{DEW}. Then we have
\begin{eqnarray*}
\Tr(\rho V_{(s)})&=&\sum_{k,l=0}^{\lfloor N(d-1)/2\rfloor}s_k\overline{s_l}\sum_{j=0}^{N(d-1)}p_j|\langle\widetilde{R_{N,d;k+l}},R_{N,d;j}\rangle|^2 \\
&=&\sum_{k,l=0}^{\lfloor N(d-1)/2\rfloor}s_k\overline{s_l}p_{k+l}
\end{eqnarray*}
Since it is nonnegative for every choice of complex numbers $(s_k)$, the Henkel matrix $(p_{k+l})_{0\leq k,l\leq \lfloor N(d-1)/2\rfloor}$ is positive semidefinite. Similarly, using the witness $U_{(t)}$ given by \eqref{Wt} one can show that the Henkel matrix $(p_{k+l+1})_{0\leq k,l\leq \lfloor (N(d-1)-1)/2\rfloor}$ is positive definite too. Hence, due to Theorem \ref{t:moment} we conclude that the sequence $(p_k)$ is a solution of the generalized moment problem.

\textit{Sufficiency.}
Assume that $(p_k)$ is a solution of the generalized moment problem. Thus, there are a positive measure $\sigma$ supported on $[0,\infty)$ and a constant $M\geq 0$ such that conditions described in \eqref{e:moment} are satisfied. For $t\geq 0$, let 
\begin{equation}
\nonumber
\rho_t=\sum_{k=0}^{N(d-1)} t^k |R_{N,d;k}\rangle\langle R_{N,d;k}|
\end{equation}
Then, it follows from \eqref{e:moment} that
\begin{equation}
\nonumber
\rho=\int_0^\infty\rho_t d\sigma(t) + M|R_{N,d;N(d-1)}\rangle\langle R_{N,d;N(d-1)}|
\end{equation} 
Observe that 
$
|R_{N,d;N(d-1)}\rangle\langle R_{N,d;N(d-1)}|=|d-1\rangle\langle d-1|^{\otimes N},
$
so it is a fully separable state.
Moreover, according to Proposition \ref{p:geom}, each $\rho_t$ is also a separable state. Consequently, $\rho$ is separable too.
\end{proof}


\section{Separability of diagonal restricted Dicke states vs PPT property}
\label{s:sepPPT}
Our aim is to prove the following main theorem
\begin{theorem}\label{t:Neven}
Assume that $d\geq 2$ is arbitrary and $N$ is even. Let $\rho$ be a state given by \eqref{rDd}. The following conditions are equivalent:
	\begin{enumerate}[(a)]
		\item $\rho$ is fully separable
		\item $\rho$ is $m$-PPT for $m=\frac{N}{2}$
		\item The sequence $(p_k)_{k=0}^{N(d-1)}$ is a solution of the generalized moment problem.
	\end{enumerate}
Moreover, if $d=2$ and $N$ is odd the above conditions are also equivalent for $m=\frac{N-1}{2}$.
\end{theorem}
Let us note	 that the above equivalence was proved for $d=2$, i.e. for qubits, in \cite{Yu2016}.

\begin{proof}[Proof of Theorem \ref{t:Neven}]
Assume that $N$ is even. Implication (a)$\Rightarrow$(b) is obvious. Implication (b)$\Rightarrow$(c) a consequence of part (i) of Theorem \ref{t:PPT} and Theorem \ref{t:moment}. The implication (c)$\Rightarrow$(a) is nothing but the ``if'' part of Theorem \ref{t:Narb}.

Now, let $d=2$, $N$ be odd and $N=2m+1$. One should show the implication (b)$\Rightarrow$(c). Due to part (ii) of Theorem \ref{t:PPT}, $m$-PPT property of $\rho$ is equivalent to positive definiteness of matrices $(p_{i+j})_{i,j=0}^m$ and $(p_{i+j+1})_{i,j=0}^m$. The rest follows from Theorem \ref{t:moment}.
\end{proof}

	On the contrary to the case $d=2$, if $N$ is odd then $\frac{N-1}{2}$-PPT property does not imply full separability of $\rho$ for $d\geq 3$. We provide the following counterexample.

Let  $N=3$ and $d=3$. 
Consider a state $\rho=\sum_{k=0}^6 p_k|D_{3,3;k}\rangle\langle D_{3,3;k}|$ where 
$$(p_k)_{k=0}^6=\left( 1,\frac{1}{4},\frac{1}{8},\frac{1}{9}, \frac{1}{8},\frac{1}{4},1\right).$$
Observe that $\rho$ is a 1-PPT state. Indeed, one can easily check that matrices
$$
\left(\begin{array}{ccc}
p_0 & p_1 & p_2 \\ p_1 & p_2 & p_3 \\ p_2 & p_3 & p_4
\end{array}\right)
\quad
\left(\begin{array}{ccc}
p_1 & p_2 & p_3 \\ p_2 & p_3 & p_4 \\ p_3 & p_4 & p_5
\end{array}\right)
\quad
\left(\begin{array}{ccc}
p_2 & p_3 & p_4 \\ p_3 & p_4 & p_5 \\ p_4 & p_5 & p_6
\end{array}\right)
$$
are positive semidefinite. According to Theorem \ref{t:PPT}(b) it follows that $\rho$ has  1-PPT property.
On the other hand one checks that the determinant of a matrix
$$
\left(\begin{array}{cccc}
p_0 & p_1 & p_2 & p_3 \\ p_1 & p_2 & p_3 & p_4 \\ p_2 & p_3 & p_4 & p_5 \\ p_3 & p_4 & p_5 & p_6
\end{array}\right)
$$
is negative, hence it is not positive semidefinite. It follows from Theorem \ref{t:Narb} that $\rho$ is not separable.



\vspace{3mm}

{\it Acknowledgments}---
AR acknowledges  grant No.2014/14/M/ST2/00818 from the National Science Center.



\begin{thebibliography}{1}
\bibitem{Yu2016}
N. Yu, Phys. Rev. A 94, 060101(R) (2016).
\bibitem{WoYe2014}
E. Wolfe and S. F. Yelin, Phys. Rev. Lett. 112, 140402
(2014).
\bibitem{KrNu1977}
M. G. Krein and A. A. Nudelman, The Markov moment
problem and extremal problems (Amer. Math. Soc., Providence,
Rhode Island, 1977).
\bibitem{EPR}
A. Einstein, B. Podolsky, and N. Rosen, Phys. Rev. 47,
777 (1935).
\bibitem{HHH1996}
M. Horodecki, P. Horodecki, and R. Horodecki, Phys.
Lett. A 283, 1 (1996).
\bibitem{Peres1996}
A. Peres, Phys. Rev. Lett. 77, 143 (1996).
\bibitem{ToGu2009}
O. Guhne and G. Toth, Phys. Rep. 474, 1 (2009).
\bibitem{Toth2010a}
G. Toth and O. Guhne, Applied Physics B 98, 617
(2010).
\bibitem{Toth2010b}
G. Toth and O. Guhne, Phys. Rev. Lett. 102, 170503
(2009).
\bibitem{Tura2018}
J. Tura, A. Aloy, R. Quesada, M. Lewenstein, and
A. Sanpera, Quantum 2, 45 (2018).
\bibitem{QRS}
R. Quesada, S. Rana, and A. Sanpera,
Phys. Rev. A 95, 042128 (2017).
\bibitem{Dicke} 
R. H. Dicke,  Phys. Rev. 93, 99--110 (1954).
\bibitem{GH}
M. Gross and S. Haroche, Phys. Rep. 93, 301--396 (1982).
\end{thebibliography}


\section*{Appendix}
\stepcounter{section}

\subsection{Proof of Proposition \ref{DEW}}
\label{a:VU}
Since $P_\mathrm{D}|\widetilde{R_{N,d;k}}\rangle =|\widetilde{R_{N,d;k}}\rangle$ for every $k$, it is clear that both $V_{(s)}$ and $U_{(t)}$ are D-symmetric. To complete the proof we show that $\Tr(V_{(s)}\rho)\geq 0$ and $\Tr(U_{(t)}\rho)\geq 0$ for every pure separable D-symmetric state $\rho$. It follows from Proposition \ref{p:Dpure} that 
$\rho=|\xi\rangle\langle\xi|^{\otimes N}$ where $\xi\in\mathbb{C}^d$ is either of the form \eqref{Dsep1} or \eqref{Dsep2}. If 
$|\xi\rangle=|d-1\rangle$, then $\rho=|d-1\rangle\langle d-1|^{\otimes N}=|\widetilde{R_{N,d;N(d-1)}}\rangle\langle \widetilde{R_{N,d;N(d-1)}}|$. Thus
\begin{eqnarray*}
\Tr(V_{(s)}\rho)
&=& \langle \widetilde{R_{N,d;N(d-1)}}|V_{(s)}|\widetilde{R_{N,d;N(d-1)}}\rangle 
\\&=&\begin{cases}
|s_{n_1}|^2 & \mbox{if $N(d-1)$ is even,} \\
0 &\mbox{if $N(d-1)$ is odd,}
\end{cases}
\end{eqnarray*}
and 
\begin{eqnarray*}
\Tr(U_{(t)}\rho)
&=& \langle \widetilde{R_{N,d;N(d-1)}}|U_{(t)}|\widetilde{R_{N,d;N(d-1)}}\rangle 
\\&=&\begin{cases}
0 & \mbox{if $N(d-1)$ is even,} \\
|s_{n_2}|^2 &\mbox{if $N(d-1)$ is odd,}
\end{cases}
\end{eqnarray*}
Notice, that in all cases we got nonnegative numbers

If $|\xi\rangle$ is of the form \eqref{Dsep2} then $\rho=C_z^N \sum_{\bi,\bj}z^{|\bi|}\overline{z}^{|\bj|}|\bi\rangle\langle\bj|$ for some $z\in\mathbb{C}$. Hence
\begin{eqnarray*}
\lefteqn{\Tr(V_{(s)}\rho)=}\\
&=& C_z^N \sum_{\bi,\bj} z^{|\bi|}\overline{z}^{|\bj|}\sum_{k,l=0}^{n_1}s_k\overline{s_l}\langle\bj| \widetilde{R_{N,d;k+l}}\rangle \langle \widetilde{R_{N,d;k+l}}|\bi\rangle 
\\&=&
C_z^N \sum_{k,l=0}^{n_1}s_k\overline{s_l}|z|^{2k+2l}\\
&=&
C_z^N \left|\sum_{k=0}^{n_1}s_k|z|^{2k}\right|^2\geq 0.
\end{eqnarray*}
Analogously, one shows that 
$$\Tr(U_{(t)}\rho)=C_z^N |z|^2\left|\sum_{k=0}^{n_2}t_k|z|^{2k}\right|^2\geq 0.$$

\subsection{PPT property for permutationally symmetric states}
\label{s:perPPT}
The aim of this subsection is to show that for permutationally symmetric states the PPT property depends only on the number of subsystems being transposed.
More precisely, let $(m_1,\ldots,m_N)$ and $(m_1',\ldots,m_N')$ be two binary systems such that $m_1+\ldots m_N=m_1'+\ldots+m_N'$. We show that a permutationally symmetric state $\rho$ is $(m_1,\ldots,m_N)$-PPT state if and only if it is $(m_1',\ldots,m_N')$-PPT. 
To this end, for a permutation $\sigma\in S_N$ denote by $F_\sigma$ an operator on $(\mathbb{C}^d)^{\otimes N}$ acting as $F_{\sigma} |\xi_1,\ldots,\xi_N\rangle=|\xi_{\sigma^{-1}(1)},\ldots, \xi_{\sigma(N)}\rangle$. 

Let $\rho$ has $(m_1,\ldots,m_N)$-PPT property. Since $m_1+\ldots m_N=m_1'+\ldots+m_N'$, there is a permutation $\sigma$ of the set $\{1,\ldots,N\}$ such that $m_k'=m_{\sigma(k)}$ for every $k=1,\ldots,N$. Since $\rho$ is permutationally symmetric, $F_\sigma\rho F_\sigma^* =\rho$. Thus, one gets
\begin{eqnarray*}
T^{m_1'}\otimes\ldots\otimes T^{m_N'}\rho
&=& T^{m_1'}\otimes\ldots\otimes T^{m_N'} F_\sigma \rho F_\sigma^* \\
&=&
T^{m_{\sigma(1)}}\otimes\ldots\otimes T^{m_{\sigma(N)}} F_\sigma\rho F_\sigma^* \\
&=&
F_{\sigma}\big[T^{m_1}\otimes\ldots\otimes T^{m_N} \rho\big]F_\sigma^*.
\end{eqnarray*}
Finally, as $T^{m_1}\otimes\ldots\otimes T^{m_N} \rho$ is positive semidefinite, $T^{m_1'}\otimes\ldots\otimes T^{m_N'}\rho$ is positive semidefinite too.
\subsection{Proof of equality \eqref{e:Al}}
\label{a:Al}
Let $\rho$ be giben by \eqref{e:rD} and $\Gamma_m$ be the partial transposition of the first $m$ subsystems. For an $N$-tuple $\bi$ let $\bi_m$ and $\bi^m$ denote respectively the $m$-tuple of first $m$ coordinates of $\bi$ and the $(N-m)$-tuple of last coordinates of $\bi$, i.e. $\bi_m=(i_1,\ldots,i_m)$ and $\bi^m=(i_{m+1},\ldots,i_N)$. Then one has
\begin{widetext}
\begin{eqnarray*}
\Gamma_m(\rho)&=&
\sum_{k=0}^{N(d-1)}p_k\;\sum_{\bi\vdash k}\; \sum_{\bj\vdash k} |\bj_m,\bi^m\rangle\langle \bi_m,\bj^m|\\
&=&
\sum_{k=0}^{N(d-1)}p_k \sum_{\substack{\bi,\bj \\ |\bi_m|+|\bj^m| =k \\ |\bj_m|+|\bi^m| =k}} |\bi_m,\bi^m\rangle\langle\bj_m,\bj^m| \\
&=&
\sum_{\substack{\bi,\bj \\ |\bi^m| - |\bi_m|=|\bj^m| -|\bj_m|}}
p_{|\bi_m|+|\bj^m|} |\bi_m,\bi^m\rangle\langle\bj_m,\bj^m| \\
&=&
\sum_{s=-m(d-1)}^{(N-m)(d-1)}\sum_{\substack{\bi_m,\bj_m \\ \max\{0,-s\}\leq |\bi_m|,|\bj_m|\leq \min\{m(d-1),(N-m)(d-1)-s\} 
}}
p_{|\bi_m|+|\bj_m|+s} |\bi_m\rangle\langle\bj_m| \otimes \\
&& \otimes \sum_{\substack{\bi^m,\bj^m \\ |\bi^m|=|\bi_m|+l \\ |\bj^m|=|\bj_m|+l }} |\bi^m\rangle\langle \bj^m|\\
&=& 
\sum_{s=-m(d-1)}^{(N-m)(d-1)}\quad\sum_{k,l=\max\{0,-s\}}^{\min\{m(d-1),(N-m)(d-1)-s\}}p_{k+l+s}\sum_{\substack{\bi_m,\bj_m \\ |\bi_m|=k \\ |\bj_m|=l}} |\bi_m\rangle\langle\bj_m| \otimes \sum_{\substack{\bi^m,\bj^m \\ |\bi_m|=k+s \\ |\bj_m|=l+s}}|\bi^m\rangle\langle \bj^m|
\\
&=&
\sum_{s=-m(d-1)}^{(N-m)(d-1)}\quad\sum_{k,l=\max\{0,-s\}}^{\min\{m(d-1),(N-m)(d-1)-s\}}p_{k+l+s}
|R_{m,d;k},R_{N-m,d;k+s}\rangle\langle R_{m,d;l},R_{N-m,d;l+s}| 
\\
&=&
\sum_{s=-m(d-1)}^{(N-m)(d-1)} A_s,
\end{eqnarray*}
where $A_s$ are defined as in \eqref{e:As}.
\end{widetext}






\end{document}